\documentclass[onecolumn,draftcls,12pt]{IEEEtran}
\usepackage{enumerate,cite,amsmath,amsfonts,amssymb,subfigure,epsfig,times,graphicx,psfrag}
\usepackage{float}
\usepackage{verbatim,color}

\usepackage{lastpage}

\usepackage{fancyhdr}

\date{}

\newcommand{\mbf}{\boldmath}
\newcommand{\smbf}{\small \boldmath}

\def\b0{{\mbox {\mbf $0$}}}

\def\bh{{\mbox {\mbf $h$}}}

\def\bq{{\mbox {\mbf $q$}}}

\def\bs{{\mbox {\mbf $s$}}}

\def\bx{{\mbox {\mbf $x$}}}
\def\bv{{\mbox {\mbf $v$}}}
\def\by{{\mbox {\mbf $y$}}}
\def\bz{{\mbox {\mbf $z$}}}

\def\bH{{\mbox {\mbf $H$}}}
\def\bU{{\mbox {\mbf $U$}}}
\def\bI{{\mbox {\mbf $I$}}}

\def\bS{{\mbox {\mbf $S$}}}

\def\sbS{{\mbox {\smbf $S$}}}

\def\bb0{{\mathbf{0}}}

\def\b_beta{\mbox{\boldmath $\beta$}}
\def\hb_beta{\mbox{\boldmath $\hat \beta$}}

\newtheorem{theorem}{Proposition}
\newtheorem{Lemma}{Lemma}
\newtheorem{Remark}{Remark}

\newtheorem{definition}{Definition}

\begin{document}
\title{On the Relationship Between the Multi-antenna Secrecy Communications and Cognitive Radio Communications}
\author{Lan Zhang, Rui Zhang, Ying-Chang Liang, Yan Xin, and Shuguang Cui
\footnote{L. Zhang is with the Department of Electrical and Computer
Engineering, National University of Singapore, Singapore 118622
(email: zhanglan@nus.edu.sg).} \footnote{R. Zhang and Y.-C. Liang
are with the Institute for Infocomm Research, A*STAR, Singapore
(emails: \{rzhang, ycliang\}@i2r.a-star.edu.sg).} \footnote{Y. Xin
is with the NEC Laboratories America, Princeton, New Jersey, USA
(email: yanxin@nec-labs.com).}\footnote{S. Cui is with the
Department of Electrical and Computer Engineering, Texas A\&M
University, College Station, Texas, USA (email: cui@ece.tamu.edu).}}

\maketitle \thispagestyle{fancy} \vspace{-0in} {\linespread{1.50}
\begin{abstract}
This paper studies the capacity of the multi-antenna or
multiple-input multiple-output (MIMO) secrecy channels with multiple
eavesdroppers having single/multiple antennas. It is known that the
MIMO secrecy capacity is achievable with the optimal transmit
covariance matrix that maximizes the minimum difference between the
channel mutual information of the secrecy user and those of the
eavesdroppers. The MIMO secrecy capacity computation can thus be
formulated as a non-convex max-min problem, which cannot be solved
efficiently by standard convex optimization techniques. To handle
this difficulty, we explore a relationship between the MIMO secrecy
channel and the recently developed MIMO cognitive radio (CR)
channel, in which the multi-antenna secondary user transmits over
the same spectrum simultaneously with multiple primary users having
single/multiple antennas, subject to the received interference power
constraints at the primary users, or the so-called ``interference
temperature (IT)'' constraints. By constructing an auxiliary CR MIMO
channel that has the same channel responses as the MIMO secrecy
channel, we prove that the optimal transmit covariance matrix to
achieve the secrecy capacity is the same as that to achieve the CR
spectrum sharing capacity with properly selected IT constraints
under certain conditions. Based on this relationship, several
algorithms are proposed to solve the non-convex secrecy capacity
computation problem by transforming it into a sequence of CR
spectrum sharing capacity computation problems that are convex. For
the case with single-antenna eavesdroppers, the proposed algorithms
obtain the exact capacity of the MIMO secrecy channel, while for the
case with multi-antenna eavesdroppers, the proposed algorithms
obtain both upper and lower bounds on the MIMO secrecy capacity.
\end{abstract}
}


%\newpage

%\pagenumbering{arabic}
\section{Introduction}

In the 1970s, Wyner introduced a secrecy transmission model in his
seminal work \cite{Wyner:wiretap75} on information-theoretic
secrecy. In this model, the secrecy transmitter sends confidential
messages to a legitimate receiver subject to the requirement that
the messages cannot be decoded by an eavesdropper. The
information-theoretic study of the secrecy transmission problem has
been continued and extended to many other channel models, including
broadcast channels (BCs) \cite{csiszar:BCCM78,ybliang:CRBC07},
multiple access channels (MACs)
\cite{ybliang:MACCM08,ybliang:MACCM06}, and interference channels
(ICs) \cite{Ruoheng08:BCICCM,ybliang:CRIC07}. Very recently, the
secrecy capacity of the multi-antenna/multiple-input multiple-output
(MIMO) channel has been characterized by Khisti and Wornell
\cite{wornell:MISOscr07}, and Oggier and Hassibi
\cite{hassibi:MIMOscr07}. In their work, the MIMO secrecy channel
with a single eavesdropper having multiple antennas is transformed
into a degraded MIMO-BC, whose capacity is an upper bound on the
secrecy capacity. It was shown in
\cite{wornell:MISOscr07,hassibi:MIMOscr07} that this capacity upper
bound is indeed tight for the Gaussian noise case, i.e., the exact
secrecy capacity. However, this computable secrecy capacity cannot
be extended to the general case of multiple eavesdroppers. Moreover,
Liu and Shammai \cite{Tieliu:scr07} also established the MIMO
secrecy capacity by using the channel enhancement technique
\cite{Shamai06:MIMOBCcapacity}. However, no computable
characterization of the secrecy capacity was provided in
\cite{Tieliu:scr07}.

On the other hand, cognitive radio is considered as an efficient
technology to dramatically improve spectrum utilization, thus having
great potential to solve spectrum scarcity problem. In a
spectrum-sharing CR system, the CR user or the so-called secondary
user (SU) is allowed to simultaneously transmit with the licensed
primary user (PU) over the same spectrum, provided that the SU to PU
interference level is regulated below a predefined threshold, which
is also called the ``interference temperature (IT)'' constraint. The
capacity achieving transmission problems under the IT constraint for
the secondary users have been studied in \cite{Liang:jstsp},
\cite{lan07:jsac}, and \cite{lan:icc08} for the CR MIMO
point-to-point channel, the CR MIMO-MAC, and the CR MIMO-BC,
respectively. Since the IT constraint is a linear function of the
transmit covariance matrix, the capacity characterization problem
for the CR MIMO channel can be formulated as a convex optimization
problem, and is thus solvable via the standard interior point method
\cite{Liang:jstsp}. It is worth noting that the system models of the
secrecy channel and the CR channel are fairly similar in the sense
that the secrecy and SU transmitters need to regulate the resultant
signal power level at the eavesdropper and PU, respectively, so as
to achieve the goals of confidential transmission and PU protection,
respectively.

In this paper, we study the capacity computation problem for the
general case of the MIMO secrecy channel with multiple eavesdroppers
having single/multiple antennas. Based on the results in
\cite{hassibi:MIMOscr07,wornell:MISOscr07}, the related MIMO secrecy
capacity can be obtained via optimizing over the transmit covariance
matrix of the secrecy user to maximize the minimum difference
between the mutual information of the secrecy channel and those of
the channels from the secrecy transmitter to different
eavesdroppers. It can thus be shown that the resulting capacity
computation problem is a non-convex max-min optimization problem,
which cannot be solved efficiently with standard convex optimization
techniques. To handle this difficulty, we consider an auxiliary CR
MIMO channel with multiple PUs having single/multiple antennas and
the same channel responses as those in the MIMO secrecy problem. We
next establish a relationship between this auxiliary CR MIMO channel
and the MIMO secrecy channel by proving that the optimal transmit
covariance matrix for the secrecy channel is the same as that for
the CR channel with properly selected IT constraints for the PUs
under certain conditions. Based on such a relationship, we transform
the non-convex MIMO secrecy capacity computation problem into a
sequence of CR capacity computation problems, which are convex and
thus can be efficiently solved. For the case of single-antenna
eavesdroppers, the proposed algorithms obtain the exact capacity of
the associated MIMO secrecy channel, while for the case of
multi-antenna eavesdroppers, the proposed algorithms obtain both the
upper and lower bounds on the MIMO secrecy capacity.

The rest of this paper is organized as follows. Section
\ref{sect:prob} presents the system models and problem formulations
for the CR MIMO transmission and the secrecy MIMO transmission.
Section \ref{sect:mnrlt} describes the main theoretical results of
this paper on the relationship between the secrecy capacity and the
CR spectrum sharing capacity. Section \ref{sect:algrm} studies the
case of single-antenna eavesdroppers, and develops several
algorithms to compute the MIMO secrecy capacity. Section
\ref{sect:m_ante} extends the results to the case of multi-antenna
eavesdroppers to obtain the upper and lower bounds on the MIMO
secrecy capacity. Section \ref{sect:sim} presents some numerical
examples. Finally, Section \ref{sect:cncl} concludes the paper.

{\it Notation}: Uppercase boldface and lowercase boldface letters
are used to denote matrices and vectors, respectively. $(\bS)^{H}$,
$\text{tr}(\bS)$, and $|\bS|$ denote the conjugate transpose, the
trace, and the determinant of a matrix $\bS$, respectively.
$\mathcal{R}^{K}$ denotes the vector space of $K\times 1$ real
vectors, and $\mathcal{R}$ denotes the field of real numbers. $\bI$
denotes an identity matrix. $\mathbb{E}[\cdot]$ denotes statistical
expectation. $|\cdot|$ denotes the absolute value of a complex
number.

\section{System Model and Problem Formulation}\label{sect:prob}

In this section, we present system models and problem formulations
for the CR MIMO transmission and the secrecy MIMO transmission in
the following two subsections, respectively.

\subsection{CR MIMO Transmission}

As shown in Fig. \ref{fig:sysmodel}(a), we consider a CR MIMO
channel, where the SU transmitter (SU-Tx) is equipped with $N$
transmit antennas, and the SU receiver (SU-Rx) is equipped with $M$
receive antennas. The SU-Tx to SU-Rx channel is denoted by a
$N\times M$ matrix $\bH_s$. Moreover, there are $K$ single-antenna
PU receivers denoted by PU$_i$, $i=1,\cdots,K$, and the channel from
SU-Tx to PU$_i$ is denoted by the $N\times 1$ vector $\bh_i$. The
received signal $\by$ at SU-Rx is expressed as
\begin{align}
\by=\bH_s^H\bx+\bz
\end{align}
where $\bx$ is the transmit signal vector at SU-Tx, and $\bz$
denotes the noise vector at SU-Rx. The entries of the noise vector
are independent circularly symmetric complex Gaussian (CSCG) random
variables of zero mean and covariance matrix $\bI$. Since the SU
shares the same spectrum with the PUs, there are $K$ IT constraints
imposed to the SU transmission, expressed as
$\mathbb{E}[|\bh_i^H\bx|^2]\le \Gamma_i, i=1,\cdots,K$, where
$\Gamma_i$ denotes the tolerable IT limit for PU$_i$.

Consider the CR MIMO transmission problem, in which we determine the
optimal transmit covariance matrix for SU-Tx to maximize the data
rate subject to the transmit power constraint and the IT constraints
for the $K$ PUs. Mathematically, this problem can be formulated as
\cite{Liang:jstsp}
\begin{align}
\begin{split}
\mathbf{(PA)}:~~~\max_{\sbS}~&\log|\bI+\bH_s^H\bS\bH_s|\notag\\
\text{subject to:}~&\text{tr}(\bS)\le P \\
&\bh_i^H\bS\bh_i\le \Gamma_i,~i=1,\cdots,K
\end{split}
\end{align}
where $\bS=\mathbb{E}[\bx\bx^H]$ denotes the transmit covariance
matrix at SU-Tx, and $P$ denotes the transmit power constraint. Note
that $\bS$ is a positive semi-definite matrix such that
$\mathbf{(PA)}$ is a convex problem and can be solved efficiently by
the standard interior point method \cite{Boyd_optimization_book}.

\subsection{Secrecy MIMO Transmission}

As shown in Fig. \ref{fig:sysmodel}(b), we consider a MIMO secrecy
channel, where the secrecy transmitter (SC-Tx) is equipped with $N$
transmit antennas, and the secrecy receiver (SC-Rx) is equipped with
$M$ receive antennas. Moreover, there are $K$ single-antenna
eavesdroppers. In accordance with the earlier introduced CR MIMO
channel, the channel response from SC-Tx to SC-Rx is denoted by
$\bH_s$, and the channel response from SC-Tx to the $i$th
eavesdropper (EA$_i$) is denoted by $\bh_i, i=1,\cdots,K$. According
to the secrecy requirement, the transmitted message $W$ from SC-Tx
should not be decoded by any of the eavesdroppers, i.e.,
$H(W|y_i)\ge r, \forall i$, with $y_i$ denoting the received signal
at EA$_i$, and $r$ denoting the secrecy transmit rate. According to
the results in \cite{wornell:MISOscr07,hassibi:MIMOscr07}, the
secrecy capacity can be obtained by solving the following
optimization problem
\begin{align}
\begin{split}
\mathbf{(PB)}:~~~\max_{\sbS}~\min_{i}~&\log|\bI+\bH_s^H\bS\bH_s|-\log\Big(1+\frac{\bh_i^H\bS\bh_i}{\sigma_i^2}\Big)\notag\\
\text{subject to:}~&\text{tr}(\bS)\le P
\end{split}
\end{align}
where $\bS$ denotes the transmit covariance matrix of SC-Tx, similar
to that of SU-Tx in the CR case, and $\sigma_i^2$ denotes the
variance of the zero-mean CSCG noise at EA$_i$.

We see that $\mathbf{(PB)}$ is a non-convex optimization problem
since its objective function is the difference between two concave
functions of $\bS$ and thus not necessarily concave. Existing
methods in the literature
\cite{Ulukus:MISOscr07,ruoheng08:MISOBCCM,wornell:MISOscr07,hassibi:MIMOscr07}
for the MIMO secrecy capacity computation is only applicable to the
case of a single eavesdropper. However, these methods cannot solve
the case with multiple eavesdroppers $\mathbf{(PB)}$ even for the
case where each eavesdropper has a single antenna\footnote{Problem
$\mathbf{(PB)}$ in the case of multi-antenna eavesdroppers will be
studied later in Section \ref{sect:m_ante}.}.

\begin{Remark}
According to Fig. \ref{fig:sysmodel}, it is easy to observe that the
system models of the CR transmission and the secrecy transmission
bear the similarity that they both need to control the received
signal power levels at both PUs and eavesdroppers. However, note
that $\mathbf{(PA)}$ guarantees that the interference power at each
PU receiver is below the required threshold without considering the
PU noise power, while for $\mathbf{(PB)}$, through the second term
in the objective function, the confidential level at each
eavesdropper is not only related to the received signal power from
SC-Tx, but also related to the noise power at eavesdroppers.
Therefore, one immediate question is whether there exists a
relationship between these two systems such that we can solve the
non-convex problem $\mathbf{(PB)}$ by transforming it into some form
of $\mathbf{(PA)}$ that is convex and thus efficiently solvable.
With this motivation, we first study the relationship between these
two problems, and then propose corresponding algorithms to solve
$\mathbf{(PB)}$.
\end{Remark}

\section{Relationship Between Secrecy Capacity and CR Spectrum Sharing Capacity}\label{sect:mnrlt}

In this section, we present main theoretical results of the paper on
the relationship between the secrecy capacity computation problem
$\mathbf{(PB)}$ and the CR spectrum sharing capacity computation
problem $\mathbf{(PA)}$. While the developed relationship applies to
both single-antenna and multi-antenna CR/secrecy channels, we are
particularly interested in the multi-antenna case since it provides
a general guidance for solving $\mathbf{(PB)}$.

\begin{theorem}\label{thm:1}
For a given $\mathbf{(PB)}$, there exists a set of IT constraint
values, $\Gamma_i$, $i=1,\cdots,K$, such that the resulting
$\mathbf{(PA)}$ has the same solution as that of $\mathbf{(PB)}$.
\end{theorem}
\begin{proof}
Please refer to Appendix \ref{app:A}.
\end{proof}

Proposition \ref{thm:1} establishes the relationship between
$\mathbf{(PA)}$ and $\mathbf{(PB)}$. To further investigate this
relationship, we define an auxiliary function of $\Gamma_i$s as
\begin{align}\label{def:g}
\begin{split}
g(\Gamma_1,\cdots,\Gamma_K):=\max_{\sbS}&~|\bI+\bH_s^H\bS\bH_s|\\
\text{subject to:}~&\text{tr}(\bS)\le P \\
&\bh_i^H\bS\bh_i\le \Gamma_i,i=1,\cdots,K.
\end{split}
\end{align}
Note that the only difference between Problem \eqref{def:g} and
$\mathbf{(PA)}$ is that the objective function in Problem
\eqref{def:g} does not involve a logarithmic function of matrix
determinant while that in (PA) does. As a result, Problem
\eqref{def:g} is non-convex since its objective function is not
concave in $\bS$. Also note that Problem \eqref{def:g} is equivalent
to $\mathbf{(PA)}$ since they have the same optimal solution for
$\bS$. Therefore, although Problem \eqref{def:g} is non-convex, its
optimal solution can be obtained via solving the convex counterpart
$\mathbf{(PA)}$.
\begin{theorem}\label{thm:quansicnv}
$\mathbf{(PB)}$ is equivalent to the following optimization problem:
\begin{align}\label{prob:frac}
\max_{\Gamma_1,\cdots,\Gamma_K}\min_i~F_i(\Gamma_1,\cdots,\Gamma_K):=\frac{g(\Gamma_1,\cdots,\Gamma_K)}{1+\Gamma_i/\sigma_i^2}.
\end{align}
\end{theorem}
\begin{proof}
Please refer to Appendix \ref{app:B}.
\end{proof}

Proposition \ref{thm:quansicnv} establishes the relationship between
$\mathbf{(PB)}$ and the auxiliary function
$g(\Gamma_1,\cdots,\Gamma_K)$ that is related to $\mathbf{(PA)}$.
The equivalence between Problem \eqref{prob:frac} and
$\mathbf{(PB)}$ means that by solving the optimal $\Gamma_i$s in
Problem \eqref{prob:frac}, we could solve an optimal $\bS$ given
$g(\Gamma_1,\cdots,\Gamma_K)$ is an embedded optimization problem
over $\bS$ inside Problem \eqref{prob:frac}. Such an optimal $\bS$
is also the solution for $\mathbf{(PB)}$, for which the explanation
is given in Appendix \ref{app:B}.

Problem \eqref{prob:frac} can be solved by utilizing an important
property of $g(\Gamma_1,\cdots,\Gamma_K)$ described as follows:
\begin{theorem}\label{thm:gcnv}
The function $g(\Gamma_1,\cdots,\Gamma_K)$ is a concave function
with respect to $\Gamma_1,\cdots,\Gamma_K$, and
\begin{align}
\gamma_i(\Gamma_1,\cdots,\Gamma_K):=\frac{\partial
g(\Gamma_1,\cdots,\Gamma_K)}{\partial
\Gamma_i}=\mu_i^{(1)}|\bI+\bH_s^H\bS^{(1)}\bH_s|,~i=1,\cdots,K
\end{align}
where $\bS^{(1)}$ and $\mu_i^{(1)}$ are the optimal solution of
$\mathbf{(PA)}$ and the corresponding Lagrange multiplier (the dual
solution) with respect to the $i$th IT constraint, respectively.
\end{theorem}
\begin{proof}
Please refer to Appendix \ref{app:C}.
\end{proof}

Note that from Proposition \ref{thm:gcnv}, it follows that the
gradient of $g(\Gamma_1,\cdots,\Gamma_K)$ in \eqref{prob:frac} can
be obtained by solving $\mathbf{(PA)}$ via the Lagrange duality
method, which completes the equivalence between $\mathbf{(PA)}$ and
$\mathbf{(PB)}$ via the intermediate problem \eqref{prob:frac}. At
last, we have
\begin{theorem}\label{thm:quscnv1}
Problem \eqref{prob:frac} is a quasi-concave maximization problem.
\end{theorem}
\begin{proof}
Please refer to Appendix \ref{app:D}.
\end{proof}
Proposition \ref{thm:quscnv1} suggests that Problem
\eqref{prob:frac} can be solved by utilizing convex optimization
techniques, for which the details are given in the next section.

\section{Algorithms}\label{sect:algrm}

In this section, we present the algorithms to compute the MIMO
secrecy capacity by exploiting the relationship between the secrecy
transmission and the CR transmission developed in Section
\ref{sect:mnrlt}. The algorithm for the general case of the MIMO
secrecy channel with multiple eavesdroppers is presented first. Two
reduced-complexity algorithms are next presented, one for the
special case with one single eavesdropper, and the other for the
special case with a single-antenna secrecy receiver, i.e., the
multiple-input single-output (MISO) secrecy channel.

\subsection{General Case}

In this subsection, we present the algorithm for $\mathbf{(PB)}$ in
the general case of MIMO secrecy channels with multiple
eavesdroppers. According to Propositions \ref{thm:quansicnv} and
\ref{thm:quscnv1}, $\mathbf{(PB)}$ is equivalent to the
quasi-concave maximization problem \eqref{prob:frac}. Thus, we
instead study Problem \eqref{prob:frac} since it is easier to handle
than $\mathbf{(PB)}$.

According to \cite{Boyd_optimization_book}, a quasi-concave
maximization problem can be reduced to solving a sequence of convex
feasibility problems. Thus, Problem \eqref{prob:frac} can be further
transformed as
\begin{align}\label{prob:fractran}
\begin{split}
\max_{t,\Gamma_1,\cdots,\Gamma_K}&~t\\
\text{subject to}:&~g(\Gamma_1,\cdots,\Gamma_K)\ge
t(1+\Gamma_i/\sigma_i^2),i=1,\cdots,K.
\end{split}
\end{align}
Let $t^*$ be the optimal solution of Problem \eqref{prob:fractran}.
Clearly, $t^*$  is also the optimal value of Problem
\eqref{prob:frac}. If the feasibility problem
\begin{align}\label{prob:feas}
\begin{split}
\max_{\Gamma_1,\cdots,\Gamma_K} &\text{0}\\
\text{subject to}:&~g(\Gamma_1,\cdots,\Gamma_K)\ge
t(1+\Gamma_i/\sigma_i^2),i=1,\cdots,K
\end{split}
\end{align}
for a given $t$ is feasible, then it follows that $t^*\ge t$.
Conversely, if Problem \eqref{prob:feas} is infeasible, then $t^*<
t$. Therefore, by assuming an interval $[~0,~\bar{t}~]$ known to
contain the optimal $t^*$, the optimal solution of Problem
\eqref{prob:fractran} can be found easily via a bisection search.
Note that a suitable value for $\bar{t}$ can be chosen as
$g(\infty,\cdots,\infty)$ from \eqref{def:g}.

We next solve the feasibility problem \eqref{prob:feas} by a similar
method discussed in \cite{rzhang06:powerregion}. It is worth noting
that the feasibility problem \eqref{prob:feas} can be viewed as an
optimization problem. The Lagrangian of Problem \eqref{prob:feas}
can be written as
\begin{align}\label{eq:Lnu}
L_0(\{\nu_i\},\Gamma_1,\cdots,\Gamma_K)=\sum_{i=1}^{K}\nu_i\Big(g(\Gamma_1,\cdots,\Gamma_K)-t(1+\Gamma_i/\sigma_i^2)\Big)
\end{align}
where $\nu_i$ is the non-negative dual variable for the $i$th
constraint, and $\{\nu_i\}$ denotes $\nu_1,\cdots,\nu_K$. The
corresponding dual function is then defined as
\begin{align}\label{eq:gnu}
f_0(\{\nu_i\})=\max_{\Gamma_1,\cdots,\Gamma_K}~\sum_{i=1}^{K}\nu_i\Big(g(\Gamma_1,\cdots,\Gamma_K)-t(1+\Gamma_i/\sigma_i^2)\Big).
\end{align}
Due to its convexity, Problem \eqref{prob:feas} can be transformed
into its equivalent dual problem as
\begin{align}\label{prob:feasdual}
\min_{\{\nu_i\}}~f_0(\{\nu_i\})
\end{align}
and the duality gap between the optimal values of Problem
\eqref{prob:feas} and Problem \eqref{prob:feasdual} is zero if
Problem \eqref{prob:feas} is feasible.

Since it is known from Proposition \ref{thm:gcnv} that function
$g(\Gamma_1,\cdots,\Gamma_K)$ is concave with respect to
$\{\Gamma_1,\cdots,\\\Gamma_K\}$, Problem \eqref{eq:gnu} can be
solved via a gradient-based algorithm. According to Proposition
\ref{thm:gcnv}, the gradient of function
$g(\Gamma_1,\cdots,\Gamma_K)$ can be obtained by solving
$\mathbf{(PA)}$. Furthermore, since function $f_0(\{\nu_i\})$ is
convex with respect to $\{\nu_i\}$, Problem \eqref{prob:feasdual}
can be solved by a subgradient-based algorithm, such as the
ellipsoid method \cite{Boyd_optimization_book}. Similar to Lemma 3.5
in \cite{rzhang06:powerregion}, Problem \eqref{prob:feas} is
infeasible if and only if there exist $\{\nu_i\}$ such that
$f_0(\{\nu_i\})<0$. Using this fact along with the subgradient-based
search over $\{\nu_i\}$, the feasibility problem \eqref{prob:feas}
can be solved. In summary, the algorithm for Problem
\eqref{prob:frac} with a target accuracy parameter $\epsilon$ is
listed as follows:

\noindent {\underline{\it Algorithm 1:}}
\begin{itemize}
    \item Initialization:
    $t^{\text{min}}=0,t^{\text{max}}=\bar{t}$.
    \item Repeat
    \begin{itemize}
        \item
        $t\leftarrow\frac{1}{2}(t^{\text{min}}+t^{\text{max}})$.
        \item Solve the feasibility problem \eqref{prob:feas}. If Problem \eqref{prob:feas}
        is feasible, $t^{\text{min}} \leftarrow t$; otherwise, $t^{\text{max}} \leftarrow t$.
        \item Stop, when $t^{\text{max}}-t^{\text{min}}\le
        \epsilon$.
    \end{itemize}
    \item The optimal value of Problem \eqref{prob:frac} is taken as $t^{\text{min}}$.
\end{itemize}

\subsection{Single-Eavesdropper Case}

We now consider a special case of $\mathbf{(PB)}$, where there is
only one single eavesdropper in the secrecy channel, and propose a
simplified algorithm over Algorithm 1 for the corresponding
$\mathbf{(PB)}$.

Consider first the counterpart CR transmission problem
$\mathbf{(PA)}$. For the single-PU case, $\mathbf{(PA)}$ can be
rewritten as
\begin{align}
\begin{split}
\mathbf{(PC)}:~~~\max_{\sbS}~&\log|\bI+\bH_s^H\bS\bH_s|\notag\\
\text{subject to:}~&\text{tr}(\bS)\le P \\
&\bh^H\bS\bh\le \Gamma
\end{split}
\end{align}
where $\bh$ denotes the channel from SU-Tx to the single PU, and
$\Gamma$ is the corresponding IT limit for the PU. On the other
hand, for the single-eavesdropper case, the secrecy transmission
problem $\mathbf{(PB)}$ can be rewritten as
\begin{align}
\begin{split}
\mathbf{(PD)}:~~~\max_{\sbS}~&\log|\bI+\bH_s^H\bS\bH_s|-\log\Big(1+\frac{\bh^H\bS\bh}{\sigma^2}\Big)\notag\\
\text{subject to:}~&\text{tr}(\bS)\le P
\end{split}
\end{align}
where $\bh$ denotes the channel from SC-Tx to the single
eavesdropper, and $\sigma^2$ denotes the variance of the noise at
the eavesdropper. Following Proposition \ref{thm:quansicnv},
$\mathbf{(PD)}$ is equivalent to the optimization problem
\begin{align}\label{prob:frac2}
\max_{\Gamma}~F(\Gamma):=\frac{g(\Gamma)}{1+\Gamma/\sigma^2}
\end{align}
where the function $g(\Gamma)$ is the single-PU counterpart of that
in \eqref{def:g}. Moreover, it is evident from Proposition
\ref{thm:quscnv1} that function $F(\Gamma)$ is quasi-concave and the
optimization problem \eqref{prob:frac2} is a quasi-concave
maximization problem.
\begin{Lemma}\label{lemma:eql}
The sufficient and necessary condition for $\Gamma^*$ to be the
optimal solution of Problem \eqref{prob:frac2} is
\begin{align}\label{eq:suffcon}
\gamma(\Gamma^*)(1+\Gamma^*/\sigma^2)=\frac{1}{\sigma^2}g(\Gamma^*)
\end{align}
where $\gamma(\Gamma):=\frac{\partial g(\Gamma)}{\partial \Gamma}$.
\end{Lemma}

\begin{proof}
Please refer to Appendix \ref{app:E}.
\end{proof}

Based on Lemma \ref{lemma:eql}, $\mathbf{(PD)}$ can be solved via
the equivalent problem \eqref{prob:frac2} by the bisection method
summarized as follows:

\noindent {\underline{\it Algorithm 2:}}
\begin{itemize}
    \item Initialization:
    $\Gamma^{\text{min}}=0,\Gamma^{\text{max}}=\bar{\Gamma}$.
    \item Repeat
    \begin{itemize}
        \item
        $\Gamma\leftarrow\frac{1}{2}(\Gamma^{\text{min}}+\Gamma^{\text{max}})$.
        \item Solve Problem \eqref{def:g} for the single-PU case, and compute
        $\gamma(\Gamma)$. If
        $\gamma(\Gamma)(1+\Gamma/\sigma^2)>\frac{1}{\sigma^2}g(\Gamma)$, $\Gamma^{\text{min}} \leftarrow \Gamma$; otherwise,
        $\Gamma^{\text{max}} \leftarrow \Gamma$.
        \item Stop, when $\Gamma^{\text{max}}-\Gamma^{\text{min}}\le
        \epsilon$.
    \end{itemize}
    \item The optimal solution of $\mathbf{(PD)}$ equals that of $\mathbf{(PC)}$ with the converged $\Gamma$.
\end{itemize}
Note that in the above algorithm, $\bar{\Gamma}=\bh^H\bS_o\bh$ and
$\bS_o$ is the optimal solution of $\mathbf{(PC)}$ without the
interference power constraint\footnote{When $\Gamma>\bar{\Gamma}$,
the value of $g(\Gamma)$ is constant regardless of $\Gamma$. Thus,
the optimal solution of Problem \eqref{prob:frac2} satisfies
$\Gamma^*\le\bar{\Gamma}$.}.

Algorithm 2 searches the optimal $\Gamma$ according to its gradient
direction, and thus avoids solving the sequence of feasibility
problems in Algorithm 1. Therefore, Algorithm 2 is much simpler than
Algorithm 1. However, since the general case of $\mathbf{(PB)}$ has
multiple variables $\Gamma_i$s, this gradient-based algorithm cannot
be applied to the general case.

\begin{Remark}
Similar to Proposition \ref{thm:1}, a dual relationship between the
secrecy transmission $\mathbf{(PD)}$ and the CR transmission
$\mathbf{(PC)}$ in the case of a single eavesdropper/PU can be
described as follows. For a given $\mathbf{(PC)}$, there is a
parameter $\sigma$, such that $\mathbf{(PD)}$ with the noise
variance $\sigma^2$ at the eavesdropper has the same solution as
that of $\mathbf{(PC)}$. This property can be proved by combining
Lemma \ref{lemma:eql} and Proposition \ref{thm:quansicnv}. This
proof is thus omitted for brevity.
\end{Remark}

\subsection{Single-Antenna SC-Rx Case}

We now turn our attention to another special case of the secrecy
channel where SC-Rx is equipped with a single receive antenna, i.e.,
the MISO secrecy channel. Same as the MIMO secrecy case, each
eavesdropper is still assumed to have a single antenna. For
notational convenience, $\mathbf{(PA)}$, $\mathbf{(PB)}$,
$\mathbf{(PC)}$,  and $\mathbf{(PD)}$  in the case of single-antenna
SC-Rx are denoted correspondingly by PA-s, PB-s, PC-s, and PD-s.

The problem PC-s has been studied in \cite{Liang:jstsp}. In
\cite{Liang:jstsp}, it was shown that the optimal transmit
covariance matrix for the CR MISO channel is a rank-one matrix, and
a closed-form solution for the optimal transmit beamforming was
presented. The problem PD-s has been studied in
\cite{Ulukus:MISOscr07,Ruoheng08:BCICCM}, where it was shown that
the optimal transmit covariance matrix for the secrecy MISO channel
is also a rank-one matrix, and based on the generalized eigenvalue
decomposition, a closed-form solution for the optimal transmit
beamforming was provided.

Consider PA-s, in which there are multiple PUs each having a single
receive antenna. To the authors' best knowledge, no closed-form
solution exists for such a case. Nevertheless, due to its convexity,
this problem can be solved via a standard interior point algorithm.
By using a similar method to that in \cite{Liang:jstsp}, it can be
shown that the optimal transmit covariance matrix for PC-s is also a
rank-one matrix. In contrast, for PB-s, due to its non-convexity,
there is no existing method in the literature to solve this problem.
However, since PB-s is a special case of $\mathbf{(PB)}$, we can
apply Algorithm 1 to efficiently solve this problem .

Next, by exploiting the special structure of PB-s, we provide a
simplified algorithm over Algorithm 1. First, we rewrite PB-s as
\begin{align}
\begin{split}
\text{(PB-s)}:~~~\max_{\sbS}\min_{i}~\hat{F}_i(\bS):=&\frac{1+\bh_s^H\bS\bh_s}{1+(\bh_i^H\bS\bh_i)/\sigma_i^2}\\
\text{subject to:}~&\text{tr}(\bS)\le P
\end{split}
\end{align}
where the $N\times 1$ vector $\bh_s$ denotes the channel from SC-Tx
to the single antenna SC-Rx. Unlike the general case of
$\mathbf{(PB)}$ where only its transformed problem in
\eqref{prob:frac} is a quasi-concave problem with respect to
$\Gamma_i$s, PB-s itself is a quasi-concave problem with respect to
$\bS$ due to the following proposition.
\begin{theorem}\label{thm:1e}
$\hat{F}_i(\bS)$ is a quasi-concave function for $i=1,\dots,K$.
\end{theorem}
\begin{proof}
Please refer to Appendix \ref{app:F}.
\end{proof}
Thus, PB-s can be transformed into the following equivalent problem
\begin{align}\label{prob:P6}
\begin{split}
\max_{\sbS,t}~&t \\
\text{subject to:}~&\text{tr}(\bS)\le P \\
&1+\bh_s^H\bS\bh_s\ge
t\Big(1+\frac{\bh_i^H\bS\bh_i}{\sigma_i^2}\Big),i=1,\cdots,K
\end{split}
\end{align}
where $t$ is a positive variable. For the fixed $t$, all the
constraints in the above problem are linear matrix inequalities over
$\bS$, and thus the corresponding feasibility problem (similarly
defined as \eqref{prob:feas}) can be viewed as a semi-definite
programming (SDP) feasibility problem. Correspondingly, the optimal
value of $t$ can be obtained by a bisection search algorithm.

Compared with Algorithm 1, the algorithm for Problem \eqref{prob:P6}
is much simpler, since the SDP feasibility problem can be solved via
high-efficiency interior point methods, while the feasibility
problem \eqref{prob:feas} in Algorithm 1 can only be solved through
a general gradient-based algorithm. Moreover, according to
Proposition \ref{thm:1}, we can find a set of parameters $\Gamma_i$s
such that the corresponding PA-s has the same solution of PB-s.
Since the optimal solution of PA-s is known to be a rank-one matrix
\cite{Liang:jstsp}, so is the optimal solution for PB-s.

\section{Multi-antenna Eavesdropper Receiver}\label{sect:m_ante}

In this section, we extend our results to the case with
multi-antenna eavesdroppers. We assume that each eavesdropper is
equipped with $N_e$ receive antennas, and the channel from SC-Tx to
the $i$th eavesdropper receiver is denoted by $\bH_i$ of size
$N\times N_e$. Similar to $\mathbf{(PB)}$, the MIMO secrecy capacity
in the multi-antenna eavesdropper case can be obtained from the
following optimization problem \cite{hassibi:MIMOscr07}
\begin{align}
\mathbf{(PE)}:~~~\max_{\sbS}\min_{i}~&\log|\bI+\bH_s^H\bS\bH_s|-\log|\bI+\bH_i^H\bS\bH_i|\label{eq:mrate}\\
\text{subject to:}~&\text{tr}(\bS)\le P.
\end{align}
To the best knowledge of the authors, there is no existing solution
in the literature for the above problem. In the following, we derive
the upper and lower bounds on the MIMO secrecy capacity in the
multi-antenna eavesdropper case based on the relationship between
the secrecy transmission and the CR transmission.

\subsection{Capacity Lower Bound}

First, we have the following lemma:

\begin{Lemma}\label{lemma:matrix inequality}
If for any $i, i\in\{1,\cdots.K\}$, $\text{tr}(\bH_i^H\bS\bH_i)\le
\Gamma_i$, we have
$|\bI+\bH_i^H\bS\bH_i|\leq(1+\frac{\Gamma_i}{L})^{L}$, where
$L=\min(N_e,N)$.
\end{Lemma}

\begin{proof}
Please refer to Appendix \ref{app:G}.
\end{proof}

Similar to Proposition \ref{thm:quansicnv}, from Lemma
\ref{lemma:matrix inequality}, the following proposition holds:
\begin{theorem}\label{thm:lowbound}
The optimal value of $\mathbf{(PE)}$ is lower-bounded by that of the
following optimization problem
\begin{align}\label{prob:frac5}
\max_{\Gamma_1,\cdots,\Gamma_K}~\min_{i}~\tilde{F}_i(\Gamma_1,\cdots,\Gamma_K):=\frac{\tilde{g}(\Gamma_1,\cdots,\Gamma_K)}{\Big(1+\frac{\Gamma_i}{L}\Big)^{L}}
\end{align}
where the function $\tilde{g}(\Gamma_1,\cdots,\Gamma_K)$ is defined
as
\begin{align}\label{def:gtilde}
\begin{split}
\tilde{g}(\Gamma_1,\cdots,\Gamma_K):=\max_{\sbS}&|\bI+\bH_s^H\bS\bH_s|\\
\text{subject to:}~&\text{tr}(\bS)\le P \\
&\text{tr}(\bH_i^H\bS\bH_i)\le \Gamma_i,~i=1,\cdots,K.
\end{split}
\end{align}
\end{theorem}
Problem (\ref{prob:frac5}) can be solved by the gradient-based
method similar to Algorithm 1. Accordingly, the lower bound on the
MIMO secrecy capacity is obtained. Note that this capacity lower
bound is tight when $N_e=1$ and thus $L=1$.

\subsection{Capacity Upper Bound}

In the multi-antenna eavesdropper case, the signals received at
different antennas of each eavesdropper are jointly processed to
decode the contained secrecy message. Therefore, a straightforward
upper bound on the secrecy capacity in this case is obtained by
assuming that the signals at different antennas of each eavesdropper
are decoded independently. Suppose that $\bh_{i,j}$ is the $j$th
column of the matrix $\bH_i, j=1,\cdots,N_e$, then the upper bound
on the secrecy capacity can be obtained as
\begin{align}
\begin{split}
\max_{\sbS}~\min_{\{i,j\}}~&\log|\bI+\bH_s^H\bS\bH_s|-\log\Big(1+\frac{\bh_{i,j}^H\bS\bh_{i,j}}{\sigma_{i,j}^2}\Big)\\
\text{subject to:}~&\text{tr}(\bS)\le P.
\end{split}
\end{align}
The above problem is the same as $\mathbf{(PB)}$ with the number of
single-antenna eavesdroppers equal to  $N_eK$, and thus can be
solved by Algorithm 1.

\section{Numerical Examples}\label{sect:sim}

In this section, we provide several numerical examples to illustrate
the effectiveness of the proposed algorithms in computing the
secrecy channel capacity under different system settings. For all
examples, we consider a MIMO secrecy channel with $M=N=4$. The
elements of the matrix $\bH_s$ and the vectors $\bh_i$s (or the
matrices $\bH_i$s in the multi-antenna eavesdropper case) are
assumed to be independent CSCG random variables of zero mean and
unit variance. Moreover, the noise power at each eavesdropper
antenna is chosen to be one, and the transmit power of the secrecy
transmitter, $P$, is defined in dB relative to the noise power.

\subsection{MIMO Secrecy Capacity with Two Single-Antenna Eavesdroppers}

In this example, we consider a MIMO secrecy channel with $K=2$
single-antenna eavesdroppers. Fig. \ref{fig:comp2met} plots the
secrecy capacity of this channel obtained by Algorithm 1, where the
transmit power ranges from 0 dB to 10 dB. Moreover, a reference
achievable secrecy rate of this channel is obtained by the
Projected-Channel SVD (P-SVD) algorithm in \cite{Liang:jstsp}. In
this algorithm, the channel $\bH_s$ is projected into a space, which
is orthogonal to $\bh_1$ and $\bh_2$, and thus the secrecy signals
cannot be received by the eavesdroppers. It is easy to observe from
Fig. \ref{fig:comp2met} that the secrecy rate obtained by P-SVD is
less than the secrecy capacity obtained by Algorithm 1. Moreover,
from Proposition \ref{thm:quscnv1}, it is known that the function
$F_i(\Gamma_1,\Gamma_2)$ is a quasi-concave function, and thus the
function $\min_{i=1,2}~F_i(\Gamma_1,\Gamma_2)$ is also a
quasi-concave function. In Fig. \ref{fig:mesh}, we plot the value of
this function for $P=5$ dB. It is observed that this function is
indeed quasi-concave.

\subsection{MIMO Secrecy Capacity with One Single-Antenna Eavesdropper}

In this example, we apply Algorithm 2 to compute the secrecy
capacity of a MIMO channel with one single-antenna eavesdropper. As
shown in Fig. \ref{fig:comp2met1}, the secrecy capacity obtained by
Algorithm 2 is larger than the achievable secrecy rate obtained by
the P-SVD algorithm. Moreover, it is verified that function
$F(\Gamma)$ defined in (\ref{prob:frac2}) is indeed quasi-concave in
Fig. \ref{fig:mesh1} for $P=5$ dB.

\subsection{MIMO Secrecy Capacity with One Multi-antenna Eavesdropper}

In this example, by applying the methods discussed in Section
\ref{sect:m_ante}, we show in Fig. \ref{fig:lowbnd} the lower and
upper bounds on the MIMO channel secrecy capacity with a single
eavesdropper using $N_e=2$ receive antennas. From the capacity lower
bound, we obtain a feasible transmit covariance matrix and thus a
corresponding achievable secrecy rate, shown in Fig.
\ref{fig:lowbnd} and named as ``Achievable Secrecy Rate''. Moreover,
the achievable secrecy rate by the P-SVD algorithm is also shown for
comparison.

\section{Conclusion}\label{sect:cncl}

In this paper, we have disclosed the relationship between the
multi-antenna CR transmission problem and the multi-antenna secrecy
transmission problem. By exploiting this relationship, we have
transformed the non-convex secrecy capacity computation problem into
a quasi-convex optimization problem, and developed various
algorithms to obtain the optimal solution for different cases of
secrecy channels. Although the proposed method cannot obtain the
exact secrecy capacity for the more complicated multi-antenna
eavesdropper case, it can be applied to compute the upper and lower
capacity bounds.

\def\appref#1{Appendix~\ref{#1}}
\appendix
\renewcommand{\thesubsection}{\Alph{subsection}}
\makeatletter
\renewcommand{\subsection}{%
\@startsection {subsection}{2}{\z@ }{2.0ex plus .5ex minus .2ex}%
{-1.0ex plus .2ex}{\it }} \makeatother

\subsection{Proof of Proposition \ref{thm:1}}\label{app:A}

Proposition \ref{thm:1} can be proved by contradiction. For the
fixed $\sigma_i$s, suppose that the optimal solution of
$\mathbf{(PB)}$ is $\bS_o$. Define
$\bar{\Gamma}_i=\bh_i^H\bS_o\bh_i, i=1,\ldots,K$. If the optimal
solution of $\mathbf{(PA)}$ with $\Gamma_i=\bar{\Gamma}_i$, $\forall
i$, denoted by $\bar{\bS}_o$, satisfies
$\log|\bI+\bH_s^H\bar{\bS}_o\bH_s|>\log|\bI+\bH_s^H\bS_o\bH_s|$,
then $\bar{\bS}_o$ is a better solution for $\mathbf{(PB)}$ than
$\bS_o$, which contradicts the preassumption that $\bS_o$ is the
optimal solution of $\mathbf{(PB)}$. Then there must be
$\log|\bI+\bH_s^H\bar{\bS}_o\bH_s|\le\log|\bI+\bH_s^H\bS_o\bH_s|$,
which means that $\bS_o$ is also the optimal solution of
$\mathbf{(PA)}$, with $\Gamma_i=\bh_i^H\bS_o\bh_i, i=1,\ldots,K$.
Proposition \ref{thm:1} thus follows.

\subsection{Proof of Proposition \ref{thm:quansicnv}}\label{app:B}

It is easy to observe that $\mathbf{(PB)}$ can be re-expressed as
\begin{align}\label{prob:P2trans}
\begin{split}
\max_{\sbS}~\min_i~&\frac{|\bI+\bH_s^H\bS\bH_s|}{1+\bh_i^H\bS\bh_i/\sigma_i^2}\\
\text{subject to:}~&\text{tr}(\bS)\le P.
\end{split}
\end{align}
Suppose that $\bS_o$ is the optimal solution of Problem
\eqref{prob:P2trans} and $\mathbf{(PB)}$. Define
$T_o:=|\bI+\bH_s^H\bS_o\bH_s|$ and
$\bar{\Gamma}_i:=\bh_i^H\bS_o\bh_i, i=1,\cdots,K$, then the optimal
objective value of Problem \eqref{prob:P2trans} is
$\bar{F}=\text{min}\Big(T_o/(1+\bar{\Gamma}_1),\cdots,T_o/(1+\bar{\Gamma}_K)\Big)$.

Suppose that the optimal solution $\bar{\bS}_o$  of Problem
\eqref{def:g} with $\Gamma_i=\bar{\Gamma}_i$, $\forall i$, satisfies
$|\bI+\bH_s^H\bar{\bS}_o\bH_s|>T_o$, then $\bar{\bS}_o$ is a better
solution for Problem \eqref{prob:P2trans} than $\bS_o$, which
contradicts the preassumption that $\bS_o$ is the optimal solution
of Problem \eqref{prob:P2trans}. Therefore, we have
$T_o=g(\bar{\Gamma}_1,\cdots,\bar{\Gamma}_K)$. Thus, $\bar{F}$ is
achievable for Problem \eqref{prob:frac} with the particular choice
of $\Gamma_i=\bar{\Gamma}_i$, $\forall i$.

Furthermore, suppose that $\tilde{\Gamma}_i$s are the optimal
solutions of Problem \eqref{prob:frac}, and the corresponding
optimal objective value is $\tilde{F}$. For Problem \eqref{def:g}
with $\Gamma_i=\tilde{\Gamma}_i$, suppose that the optimal solution
is $\tilde{\bS}$. We can prove that $\tilde{F}\le \bar{F}$ by
contradiction: If $\tilde{F}> \bar{F}$, $\tilde{\bS}$ is a better
solution for Problem \eqref{prob:P2trans} than $\bS_o$, which
contradicts the preassumption that $\bS_o$ is the optimal solution
of Problem \eqref{prob:P2trans}. As such, we see that $\bar{F}$ is
not only achievable for Problem \eqref{prob:frac}, but also the
optimal value of Problem \eqref{prob:frac} with the optimal
solutions given as $\tilde{\bS}=\bS_o$ and
$\tilde{\Gamma}_i=\bh_i^H\bS_o\bh_i$, $\forall i$ (Note that $\bS$
is a hidden design variable for Problem \eqref{prob:frac}.).
Proposition \ref{thm:quansicnv} thus follows.

\subsection{Proof of Proposition \ref{thm:gcnv}}\label{app:C}
We first study several important properties of Problem \eqref{def:g}
that is known to be an equivalent problem of $\mathbf{(PA)}$.
Considering $\mathbf{(PA)}$ first, its Lagrangian function can be
written as
\begin{align}
L_1(\bS,\lambda,\{\mu_i\})=\log|\bI+\bH_s^H\bS\bH_s|-\lambda(\text{tr}(\bS)-P)-\sum_{i=1}^{K}\mu_i(\bh_i^H\bS\bh_i-\Gamma_i)
\end{align}
where $\lambda$ and $\mu_i$ are the non-negative Lagrange
multipliers/dual variables with respect to the transmit power
constraint and the interference power constraint at PU$_i$,
respectively. Since $\mathbf{(PA)}$ is a convex optimization
problem, the Karush-Kuhn-Tucker (KKT) conditions
\cite{Boyd_optimization_book} are both sufficient and necessary for
a solution to be optimal, and solving $\mathbf{(PA)}$ is equivalent
to solving its dual problem
\begin{align}\label{prob:L1}
\min_{\lambda,\{\mu_i\}}\max_{\sbS}L_1(\bS,\lambda,\{\mu_i\}).
\end{align}

On the other hand, the auxiliary problem \eqref{def:g} is non-convex
due to the fact that its objective function is not concave. In
general, the KKT conditions may not be sufficient for a feasible
solution to be optimal when we have a non-convex optimization
problem. However, we prove in the following lemma that this is not
the case for Problem \eqref{def:g}.

\begin{Lemma}\label{lemma:dualprb}
With Problem \eqref{def:g}, the KKT conditions are both sufficient
and necessary for a solution to be optimal.
\end{Lemma}

\begin{proof}
The necessary part of Lemma \ref{lemma:dualprb} is obvious even for
a non-convex optimization problem \cite{Boyd_optimization_book}. The
sufficient part of Lemma \ref{lemma:dualprb} can be proved via
contradiction as follows. The Lagrangian of Problem \eqref{def:g}
can be written as
\begin{align}\label{eq:L2}
L_2(\bS,\delta,\{\gamma_i\})=|\bI+\bH_s^H\bS\bH_s|-\delta(\text{tr}(\bS)-P)-\sum_{i=1}^{K}\gamma_i(\bh_i^H\bS\bh_i-\Gamma_i)
\end{align}
where $\delta$ and $\gamma_i$ are the non-negative dual variables
with respect to the transmit power constraint and the interference
power constraint at PU$_i$, respectively. We first list the KKT
conditions of Problem \eqref{def:g} as follows:
\begin{align}
|\bI+\bH_s^H\bS\bH_s|\bH_s(\bI+\bH_s^H\bS\bH_s)^{-1}\bH_s^H&=\delta\bI+\sum_{i=1}^{K}\gamma_i\bh_i\bh_i^H\label{eq:kktpa1}\\
\delta(\text{tr}(\bS)-P)&=0\label{eq:kktpa2}\\
\gamma_i(\bh_i^H\bS\bh_i-\Gamma_i)&=0,
~i=1,\cdots,K.\label{eq:kktpa3}
\end{align}
Suppose that $\bS^{(0)}$, $\delta^{(0)}$, and $\gamma_i^{(0)}$ are a
set of primal and dual variables that satisfy the above KKT
conditions, and the corresponding optimal value of Problem
\eqref{def:g} is $C^{(0)}$.

The KKT conditions of $\mathbf{(PA)}$ are expressed as
\begin{align}
\bH_s(\bI+\bH_s^H\bS\bH_s)^{-1}\bH_s^H&=\lambda\bI+\sum_{i=1}^{K}\mu_i\bh_i\bh_i^H \\
\lambda(\text{tr}(\bS)-P)&=0 \\
\mu_i(\bh_i^H\bS\bh_i-\Gamma_i)&=0,~i=1,\cdots,K.
\end{align}
Suppose that $\bS^{(1)}$, $\lambda^{(1)}$, and $\mu_i^{(1)}$ are the
optimal primal and dual variables that satisfy the above KKT
conditions, and the corresponding optimal value of $\mathbf{(PA)}$
is $C^{(1)}$. Note that since $\mathbf{(PA)}$ is convex, the KKT
conditions are both necessary and sufficient.

If \eqref{eq:kktpa1}-\eqref{eq:kktpa3} are not sufficient such that
$\log(C^{(0)})\neq C^{(1)}$, i.e., $\bS^{(0)}\neq \bS^{(1)}$, we can
choose
\begin{align}
\bS&=\bS^{(0)} \label{eq:kktp2}\\
\lambda&=\delta^{(0)}/|\bI+\bH_s^H\bS^{(0)}\bH_s|\\
\mu_i&=\gamma_i^{(0)}/|\bI+\bH_s^H\bS^{(0)}\bH_s|, ~ i=1,\cdots,K
\label{eq:gamma}
\end{align}
for $\mathbf{(PA)}$, which clearly also satisfy the KKT conditions
of $\mathbf{(PA)}$. Given the sufficiency of the KKT conditions for
$\mathbf{(PA)}$, $\bS^{(0)}$ is also optimal for $\mathbf{(PA)}$
based on \eqref{eq:kktp2} such that $\log(C^{(0)})= C^{(1)}$, which
contradicts our assumption that $\log(C^{(0)})\neq C^{(1)}$. Lemma
\ref{lemma:dualprb} thus follows.
\end{proof}

Essentially, it is due to the equivalence between the non-convex
Problem \eqref{def:g} and the convex $\mathbf{(PA)}$ that Lemma
\ref{lemma:dualprb} holds. From Lemma \ref{lemma:dualprb}, it
follows that the duality gap between Problem \eqref{def:g} and its
dual problem, defined as
\begin{align}\label{prob:L2}
D=\min_{\delta,\{\gamma_i\}}\max_{\sbS}~L_2(\bS,\delta,\{\gamma_i\}),
\end{align}
is zero, i.e.,
$g(\Gamma_1,\cdots,\Gamma_K)=\min_{\delta,\{\gamma_i\}}\max_{\sbS}~L_2(\bS,\delta,\{\gamma_i\})$.
As such, from \eqref{eq:L2} we have
\begin{align}\label{eq:partgamma}
\frac{\partial g(\Gamma_1,\cdots,\Gamma_K)}{\partial
\Gamma_i}=\frac{\partial D}{\partial \Gamma_i}=\gamma_i^{(0)},
i=1,\cdots,K.
\end{align}
%Furthermore, due to the Sion's minimax theorem
%\cite{Sion:generalminimax58}, we have
%$g(\Gamma)=\max_{\sbS}\min_{\delta,\gamma}~L2(\bS,\delta,\gamma)$.
Combining \eqref{eq:gamma} and \eqref{eq:partgamma}, the latter part
of Proposition \ref{thm:gcnv} thus follows.

Now we prove the concavity of $g(\Gamma_1,\cdots,\Gamma_K)$. For the
function $g(\bq)$, where
$\bq:=[\Gamma_1,\cdots,\Gamma_K]^T\in\mathcal{R}^K_+$, its concavity
can be verified by considering an arbitrary line given by
$\bq=\bx+t\bv$, where $\bx\in \mathcal{R}^{K}_+$, $\bv\in
\mathcal{R}^K$, $t\in \mathcal{R}_+$, and $\bx+t\bv\in
\mathcal{R}^{K}_+$\cite{Boyd_optimization_book}. In the sequel, we
just need to prove that the function $g(\bx+t\bv)$ is concave with
respect to $t$. Moreover, if the $i$th IT constraint is not active
for Problem \eqref{def:g}, we have $\gamma_i=0$ from the KKT
condition such that the concavity holds. To exclude the above
trivial case, we assume that all $K$ IT constraints are active for
Problem \eqref{def:g} in the following.

Define
\begin{align}\label{def:f}
f_2(\delta,\gamma_1,\cdots,\gamma_K):=\max_{\sbS}L_2(\bS,\delta,\gamma_1,\cdots,\gamma_K)
\end{align}
as the dual function of Problem \eqref{def:g}. Let $\bs$ be the
subgradient of $f_2(\delta,\gamma_1,\cdots,\gamma_K)$. According to
the definition of subgradient, the subgradient at the point
$[\tilde{\delta},\tilde{\gamma}_1,\cdots,\tilde{\gamma}_K]$
satisfies
\begin{align}\label{def:subg}
f_2(\bar{\delta},\bar{\gamma}_1,\cdots,\bar{\gamma}_K)\geq
f_2(\tilde{\delta},\tilde{\gamma}_1,\cdots,\tilde{\gamma}_K)+([\bar{\delta},\bar{\gamma}_1,\cdots,\bar{\gamma}_K]-[\tilde{\delta},\tilde{\gamma}_1,\cdots,\tilde{\gamma}_K])\cdot
\bs,
\end{align}
where $[\bar{\delta},\bar{\gamma}_1,\cdots,\bar{\gamma}_K]$ is
another arbitrary feasible point.

\begin{Lemma}\label{lemma:muincr}
The subgradient $\bs$ of function
$f_2(\delta,\gamma_1,\cdots,\gamma_K)$ at point
$[\tilde{\delta},\tilde{\gamma}_1,\cdots,\tilde{\gamma}_K]$ is
$[P-\text{tr}(\tilde{\bS}),\Gamma_1-\bh_1^H\tilde{\bS}\bh_1,\cdots,\Gamma_K-\bh_K^H\tilde{\bS}\bh_K]$,
where $\tilde{\bS}$ is the optimal solution of Problem \eqref{def:f}
at this point.
\end{Lemma}

\begin{proof}
Let $\bar{\bS}$ be the optimal solution of Problem \eqref{def:f}
with $\delta=\bar{\delta}$ and $\gamma_i=\bar{\gamma}_i,
i=1,\cdots,K$. Thus, we have
\begin{align}
f_2(\bar{\delta},\bar{\gamma}_1,\cdots,\bar{\gamma}_K)&=\bar{r}-\bar{\delta}(\text{tr}(\bar{\bS})-P)-\sum_{i=1}^{K}\bar{\gamma}_i(\bh_i^H\bar{\bS}\bh_i-\Gamma_i)\\
&\ge
\tilde{r}-\bar{\delta}(\text{tr}(\tilde{\bS})-P)-\sum_{i=1}^K\bar{\gamma}_i(\bh_i^H\tilde{\bS}\bh_i-\Gamma_i)
\end{align}
\begin{align}
&=
\tilde{r}-\tilde{\delta}(\text{tr}(\tilde{\bS})-P)-\sum_{i=1}^K\tilde{\gamma}_i(\bh_i^H\tilde{\bS}\bh_i-\Gamma_i)+\tilde{\delta}(\text{tr}(\tilde{\bS})-P)+\sum_{i=1}^K\tilde{\gamma}_i(\bh_i^H\tilde{\bS}\bh_i-\Gamma_i) \nonumber \\
&~~~-\bar{\delta}(\text{tr}(\tilde{\bS})-P)-\sum_{i=1}^{K}\bar{\gamma}_i(\bh_i^H\tilde{\bS}\bh_i-\Gamma_i)\\
&=f_2(\tilde{\delta},\bar{\gamma}_1,\cdots,\bar{\gamma}_K)+(\text{tr}(\tilde{\bS})-P)(\tilde{\delta}-\bar{\delta})+\sum_{i=1}^K(\bh_i^H\tilde{\bS}\bh_i-\Gamma_i)(\tilde{\gamma}_i-\bar{\gamma}_i)
\label{eq:subgra}
\end{align}
where $\bar{r}=|\bI+\bH_s^H\bar{\bS}\bH_s|$ and
$\tilde{r}=|\bI+\bH_s^H\tilde{\bS}\bH_s|$. According to
\eqref{eq:subgra}, we have Lemma \ref{lemma:muincr}.
\end{proof}

According to Lemma \ref{lemma:dualprb}, Problem \eqref{def:g} is
equivalent to its dual problem
\begin{align}
\min_{\delta,\gamma_1,\cdots,\gamma_K}~f_2(\delta,\gamma_1,\cdots,\gamma_K)
\end{align}
where $f_2(\delta,\gamma_1,\cdots,\gamma_K)$ is convex. We next
consider Problem \eqref{def:g} with parameters
$P,\Gamma_1,\cdots,\Gamma_K$, denoted as Problem I. Assume that
$\bS^{(1)}$, $\delta^{(1)},\gamma_1^{(1)},\cdots,\gamma_K^{(1)}$ are
its primal and dual optimal solutions. Moreover, we have another
form of Problem \eqref{def:g} with parameters
$P,\Gamma_1+tv_1,\cdots,\Gamma_K+tv_K$, denoted as Problem II, where
$t$ is a positive constant and $v_i$ is a real constant. Assume that
$\bS^{(2)}$, $\delta^{(2)},\gamma_1^{(2)},\cdots,\gamma_K^{(2)}$ are
the primal and dual optimal solutions of Problem II. According to
\eqref{def:f}, we can write the dual function of Problem II as
\begin{align}
f_2^{\text{II}}(\delta,\gamma_1,\cdots,\gamma_K):=\max_{\sbS}~|\bI+\bH_s^H\bS\bH_s|-\delta\big(\text{tr}(\bS)-P\big)-\sum_{i=1}^{K}\gamma_i(\bh_i^H\bS\bh_i-\Gamma_i-tv_i)
\end{align}

To solve Problem II, we apply the subgradient-based algorithm to
search the minimum of its dual function
$f_2^{\text{II}}(\delta,\gamma_1,\cdots,\gamma_K)$ along the
subgradient direction. Suppose that we start from the point
$[\delta^{(1)},\gamma_1^{(1)},\cdots,\\\gamma_K^{(1)}]$. Based on
Lemma \ref{lemma:muincr}, one valid subgradient of
$f_2(\delta,\gamma_1,\cdots,\gamma_K)$ at this point is
\begin{align}
[0,\Gamma_1+tv_1-\bh_1^H\bS^{(1)}\bh_1,\cdots,\Gamma_K+tv_K-\bh_K^H\bS^{(1)}\bh_K]
=[0,tv_1,\cdots,tv_K],\label{eq:subg}
\end{align}
where \eqref{eq:subg} is due to the KKT condition of Problem I:
$\Gamma_i^{(1)}-\bh_i^H\bS^{(1)}\bh_i=0$ given $\gamma_i^{(1)}>0,
\forall i$. Moreover, according to \eqref{def:subg}, we have
\begin{align}\label{eq:subg1}
f_2^{\text{II}}(\delta^{(2)},\gamma_1^{(2)},\!\cdots\!,\gamma_K^{(2)})\ge
f_2^{\text{II}}(\delta^{(1)},\gamma_1^{(1)},\!\cdots\!,\gamma_K^{(1)})+([\delta^{(2)},\gamma_1^{(2)},\!\cdots\!,\gamma_K^{(2)}]-[\delta^{(1)},\gamma_1^{(1)},\cdots,\gamma_K^{(1)}])\cdot
\bs^{(1)},
\end{align}
where $\bs^{(1)}$ is the subgradient at the point
$[\delta^{(1)},\gamma_1^{(1)},\cdots,\gamma_K^{(1)}]$. Since
$\delta^{(2)},\gamma_1^{(2)},\cdots,\gamma_K^{(2)}$ are the dual
optimal solutions of Problem II, we have
$f_2^{\text{II}}(\delta^{(2)},\gamma_1^{(2)},\cdots,\gamma_K^{(2)})\le
f_2^{\text{II}}(\delta^{(1)},\gamma_1^{(1)},\cdots,\gamma_K^{(1)})$.
Combining this with \eqref{eq:subg} and \eqref{eq:subg1}, we have
\begin{align}
\sum_{i=1}^{K}\gamma_i^{(2)}tv_i\le
\sum_{i=1}^{K}\gamma_i^{(1)}tv_i.
\end{align}
Thus,
\begin{align}\label{ieq:1}
\sum_{i=1}^{K}\gamma_i^{(2)}v_i\le \sum_{i=1}^{K}\gamma_i^{(1)}v_i,~
\text{given}~t>0 .
\end{align}

Moreover, according to Lemma \ref{lemma:dualprb} and \eqref{eq:L2},
we have
\begin{align}\label{eq:ggra}
\frac{\partial g(\bx+t\bv)}{\partial t}=\sum_{i=1}^{K}\gamma_iv_i.
\end{align}
Note that $\gamma_i$ is the Lagrange multiplier of Problem
\eqref{def:g} with respect to the $i$th IT constraint. With a
different IT threshold, i.e., a different value of $t$, $\gamma_i$s
are not necessarily the same, and thus $\gamma_i$s can be viewed as
implicit functions of $t$. Combining \eqref{ieq:1} with
\eqref{eq:ggra}, it is easy to observe $\frac{\partial
g(\bx+t\bv)}{\partial t}$ decreases with the increase of $t$ since
the derivative change over $t$ is given as
$\sum_{i=1}^{K}\gamma_i^{(2)}v_i- \sum_{i=1}^{K}\gamma_i^{(1)}v_i\le
0$, i.e., the second order derivative of function $g(\bx+t\bv)$ over
$t$ is negative on an arbitrary line $\bx+t\bv$ in the feasible
region. Therefore, $g(\bq)$ is concave. Proposition \ref{thm:gcnv}
thus follows.

\subsection{Proof of Proposition \ref{thm:quscnv1}}\label{app:D}

The quasi-concavity is define as follows
\cite{Boyd_optimization_book}:
\begin{definition}
A function $f:\mathcal{R}^K\rightarrow\mathcal{R}$ is called
\emph{quasi-concave} if all its sublevel sets
\begin{align}
S_{\alpha}=\{\bx\in \textbf{dom}f|f(\bx)\ge\alpha\}
\end{align}
for $\alpha\in\mathcal{R}$, are convex sets.
\end{definition}
According to Proposition \ref{thm:gcnv},
$g(\Gamma_1,\cdots,\Gamma_K)$ is a concave function of $\Gamma_i$s.
Therefore, the $\alpha$-sublevel set of
$F_i(\Gamma_1,\cdots,\Gamma_K)$
\begin{align}
S_{\alpha}=\Big\{\bq\Big|\frac{g(\Gamma_1,\cdots,\Gamma_K)}{1+\Gamma_i/\sigma_i^2}\ge
\alpha\Big\}=\{\bq|g(\Gamma_1,\cdots,\Gamma_K)\ge\alpha(1+\Gamma_i/\sigma_i^2)\}
\end{align}
is a convex set for any $\alpha$, and thus the function
$F_i(\Gamma_1,\cdots,\Gamma_K)$ is a quasi-concave function. Since
the objective function of Problem \eqref{prob:frac} is the minimum
of $K$ quasi-concave functions, $F_i(\Gamma_1,\cdots,\Gamma_K)$,
$i=1,\cdots,K$, it is still quasi-concave
\cite{Boyd_optimization_book}. Proposition \ref{thm:quscnv1} thus
follows.

\subsection{Proof of Lemma \ref{lemma:eql}}\label{app:E}
The optimality condition of Problem \eqref{prob:frac2} is
\begin{align}
\frac{\partial F(\Gamma)}{\partial
\Gamma}=\frac{\gamma(\Gamma)(1+\Gamma/\sigma^2)-\frac{1}{\sigma^2}g(\Gamma)}{\big(1+\Gamma/\sigma^2\big)^2}=0.
\end{align}
Since the above optimality condition is a necessary condition for
any unconstrained smooth optimization problems regardless of its
convexity \cite{Boyd_optimization_book}, the necessary part of Lemma
\ref{lemma:eql} follows.

We next prove the sufficient part of this lemma by contradiction. We
first present a property of $\gamma(\Gamma)$ as follows.
\begin{Lemma}\label{lemma:gammadec}
$\gamma(\Gamma)$ is a non-increasing function for $\Gamma\ge0$.
\end{Lemma}

The proof of Lemma \ref{lemma:gammadec} is similar to that of
Proposition \ref{thm:quscnv1}, and thus is omitted here. Suppose
that there are two solutions $\Gamma^{(1)}$ and $\Gamma^{(2)}$, both
of which satisfy the condition in \eqref{eq:suffcon}. Furthermore,
without loss of generality, we assume $\Gamma^{(1)}<\Gamma^{(2)}$.
Therefore, we have
\begin{align}\label{eq:gGamma}
g(\Gamma^{(1)})>g(\Gamma^{(2)}).
\end{align}

According to \eqref{eq:partgamma}, we have
\begin{align}
\frac{g(\Gamma^{(2)})-g(\Gamma^{(1)})}{\Gamma^{(2)}-\Gamma^{(1)}}\ge
\gamma(\Gamma^{(2)}).
\end{align}
Thus,
\begin{align}
g(\Gamma^{(2)})-\Gamma^{(2)}\gamma(\Gamma^{(2)})\ge g(\Gamma^{(1)})-
\Gamma^{(1)}\gamma(\Gamma^{(2)})\ge g(\Gamma^{(1)})-
\Gamma^{(1)}\gamma(\Gamma^{(1)})\label{eq:proof1}
\end{align}
where the second inequality is due to the fact that
$\gamma(\Gamma^{(1)})\ge \gamma(\Gamma^{(2)})$ from Lemma
\ref{lemma:gammadec}.

Since both solutions satisfy the necessary condition
\eqref{eq:suffcon}, we have
\begin{align}
\gamma(\Gamma^{(1)})&=\frac{1}{\sigma^2}\Big(g(\Gamma^{(1)})-\gamma(\Gamma^{(1)})\Gamma^{(1)}\Big)\label{eq:proof2}\\
\gamma(\Gamma^{(2)})&=\frac{1}{\sigma^2}\Big(g(\Gamma^{(2)})-\gamma(\Gamma^{(2)})\Gamma^{(2)}\Big).\label{eq:proof3}
\end{align}
From \eqref{eq:proof1} and $\gamma(\Gamma^{(1)})\ge
\gamma(\Gamma^{(2)})$, it is easy to observe that \eqref{eq:proof2}
and \eqref{eq:proof3} hold simultaneously if and only if
$\gamma(\Gamma^{(1)})= \gamma(\Gamma^{(2)})=\gamma$. Thus, we have
\begin{align}
g(\Gamma^{(1)})=g(\Gamma^{(2)})-\gamma(\Gamma^{(1)}-\Gamma^{(2)}).
\end{align}
Since $\Gamma^{(1)}<\Gamma^{(2)}$ and
$g(\Gamma^{(1)})>g(\Gamma^{(2)})$, we further derive $\gamma<0$,
which contradicts the fact that the Lagrange multiplier
$\gamma\ge0$. As such the solution of \eqref{eq:suffcon} is unique,
which implies the sufficiency given the already proven necessity
part. Lemma \ref{lemma:eql} thus follows.

\subsection{Proof of Proposition \ref{thm:1e}}\label{app:F}
Similar to the proof given in Appendix \ref{app:D}, the
$\alpha$-sublevel set of $\hat{F}_i(\bS)$
\begin{align}
S_{\alpha}&=\{\bS|\frac{1+\bh_s^H\bS\bh_s}{1+(\bh_i^H\bS\bh_i)/\sigma_i^2}\ge
\alpha\}\\
&=\{\bS|1+\bh_s^H\bS\bh_s\ge
\alpha(1+(\bh_i^H\bS\bh_i)/\sigma_i^2)\}.
\end{align}
is a convex set. Thus, $\hat{F}_i(\bS)$ is a quasi-concave function.

\subsection{Proof of Lemma \ref{lemma:matrix inequality}}\label{app:G}
We have
\begin{align}\label{eq:1}
|\bI+\bH_i^H\bS\bH_i|=|\bI+\bU_i^H\mathbf{\Lambda}_i\bU_i|=|\bI+\mathbf{\Lambda}_i|
\end{align}
where $\bH_i^H\bS\bH_i:=\bU_i^H\mathbf{\Lambda}_i\bU_i$ is the
eigenvalue decomposition.  Since
$\text{tr}(\bH_i^H\bS\bH_i)=\text{tr}(\mathbf{\Lambda}_i)$, from
$\text{tr}(\bH_i^H\bS\bH_i)\le \Gamma_i$ it follows that
\begin{align}\label{eq:2}
\text{tr}(\mathbf{\Lambda}_i)\le \Gamma_i.
\end{align}
Combining \eqref{eq:1} and \eqref{eq:2} and denoting
$L=\min(N_e,N)$, we have
\begin{align}
|\bI+\bH_i^H\bS\bH_i|\le
\big|\bI+\frac{\Gamma_i}{L}\bI\big|=\big(1+\frac{\Gamma_i}{L}\big)^{L}
\end{align}
where the inequality is obtained by solving the following problem:
$\max_{\text{tr}(\mathbf{\Lambda}_i)\le\Gamma_i}|\bI+\mathbf{\Lambda}_i|$.
Lemma \ref{lemma:matrix inequality} thus follows.

\linespread{1.5}

\newpage

\linespread{1.6}

\begin{figure}
    \begin{minipage}[b]{1.0\linewidth}
  \centering
      \psfrag{hc}{$\bH_s$}
      \psfrag{h1}{$\bh_1$}
      \psfrag{h2}{$\bh_K$}
      \psfrag{PU1}{$\text{PU}_1$}
      \psfrag{PU2}{$\text{PU}_K$}
      \psfrag{SUtx}{$\text{SU-Tx}$}
      \psfrag{SUrx}{$\text{SU-Rx}$}
  \includegraphics[width = 100mm]{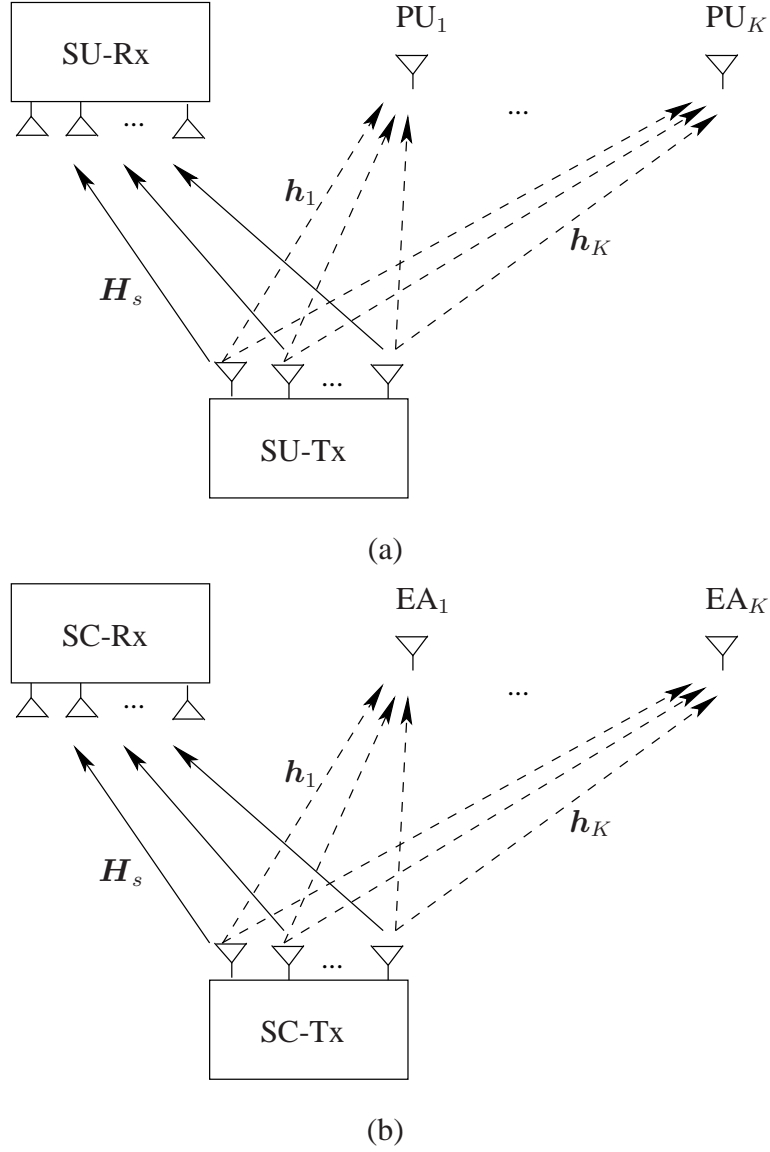}
        \centerline{(a)}\medskip
    \end{minipage}
\begin{minipage}[b]{1.0\linewidth}
  \centering
      \psfrag{hc}{$\bH_s$}
      \psfrag{h1}{$\bh_1$}
      \psfrag{h2}{$\bh_K$}
      \psfrag{PU1}{$\text{EA}_1$}
      \psfrag{PU2}{$\text{EA}_K$}
      \psfrag{SUtx}{$\text{SC-Tx}$}
      \psfrag{SUrx}{$\text{SC-Rx}$}
  \includegraphics[width = 100mm]{syscrst.eps}
        \centerline{(b)}\medskip
    \end{minipage}
     \caption{\small The comparison of two system models:
     (a) the CR MIMO channel with $K$ single-antenna PUs;
     and (b) the secrecy MIMO channel with $K$ single-antenna eavesdroppers.}
     \label{fig:sysmodel}
\end{figure}

\begin{figure}
     \centering
     \includegraphics[width = 100mm]{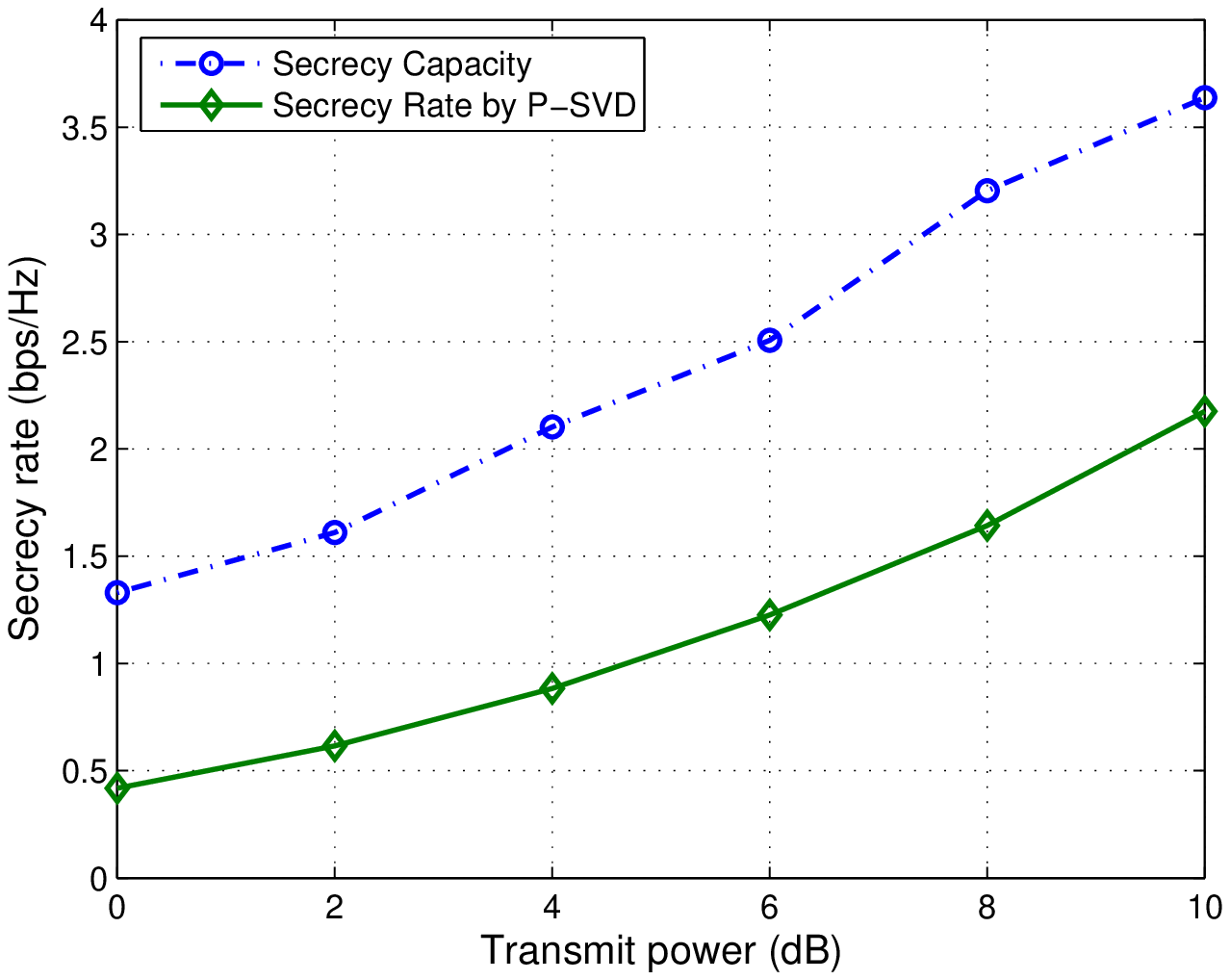}
     \caption{\small Comparison of the secrecy capacity by Algorithm 1 and the secrecy rate by the P-SVD
     algorithm in \cite{Liang:jstsp} for $M=N=4$ and $K=2$ single-antenna eavesdroppers.}
     \label{fig:comp2met}
\end{figure}

\begin{figure}
     \centering
     \includegraphics[width = 100mm]{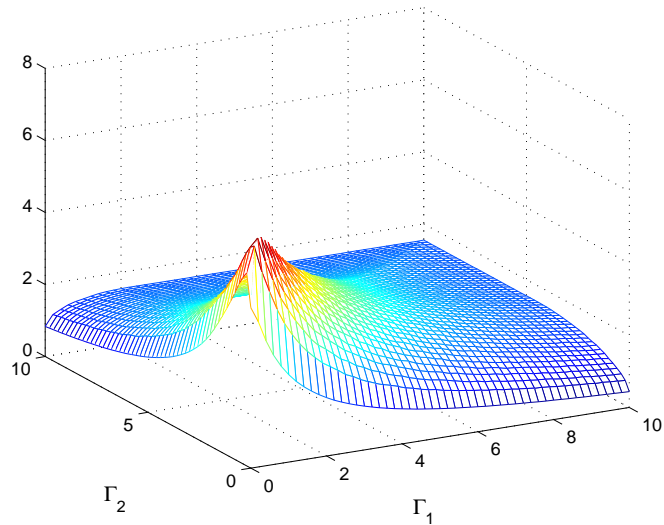}
     \caption{\small The value of the function $\min_{i=1,2}~F_i(\Gamma_1,\Gamma_2)$ for $M=N=4$, $K=2$ single-antenna eavesdroppers, and $P=5$ dB.}
     \label{fig:mesh}
\end{figure}

\begin{figure}
     \centering
     \includegraphics[width = 100mm]{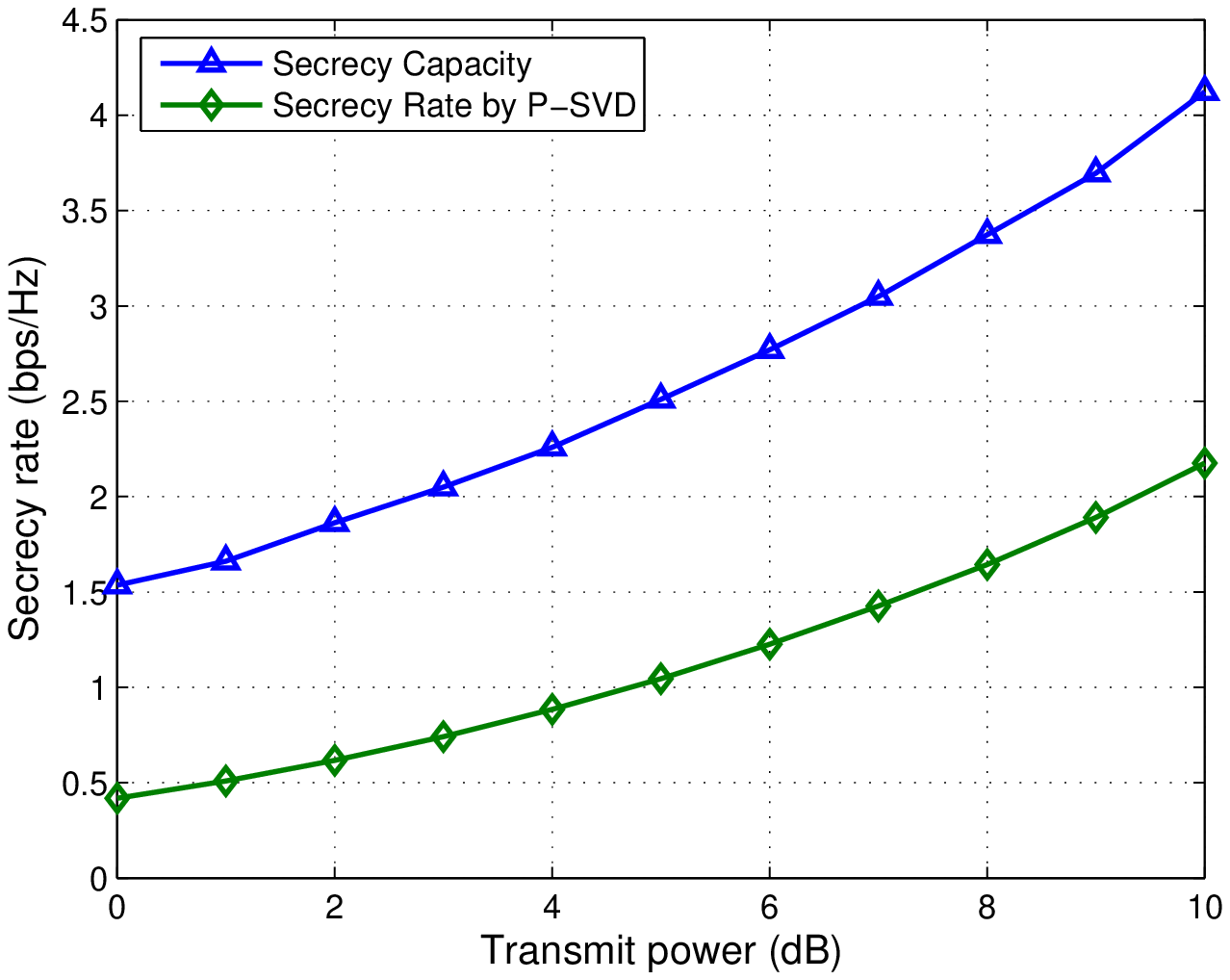}
     \caption{\small Comparison of the secrecy capacity by Algorithm 2 and the secrecy rate by the P-SVD algorithm in \cite{Liang:jstsp}
      for $M=N=4$ and $K=1$ single-antenna eavesdropper.}
     \label{fig:comp2met1}
\end{figure}

\begin{figure}
     \centering
     \includegraphics[width = 100mm]{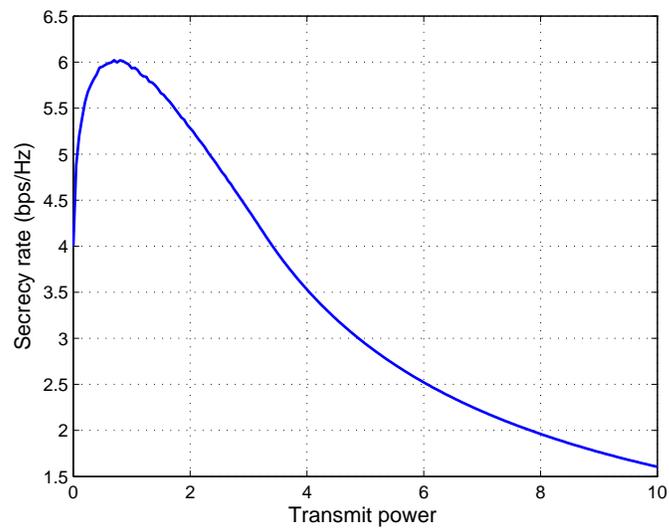}
     \caption{\small The value of the function $F(\Gamma)$ for $M=N=4$, $K=1$ single-antenna eavesdropper, and $P=5$ dB.}
     \label{fig:mesh1}
\end{figure}

\begin{figure}
     \centering
     \includegraphics[width = 100mm]{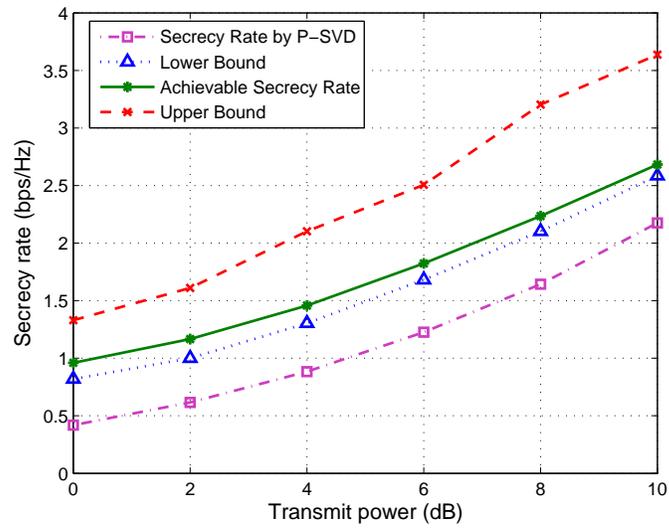}
     \caption{\small Comparison of the lower and upper bounds on the secrecy capacity and two achievable secrecy rates for $M=N=4$, $K=1$ eavesdropper
     with $N_e=2$ receive antennas.}
     \label{fig:lowbnd}
\end{figure}

\end{document}